\documentclass[%
 reprint,
 amsmath,amssymb,
 aps,
]{revtex4-1}

\usepackage{graphicx}
\usepackage{dcolumn}
\usepackage{bm}

\usepackage[utf8]{inputenc}
\usepackage[english]{babel}
\usepackage{amsthm}
\usepackage{subfigure}

\newtheorem{claim}{Claim}[section]

\DeclareMathOperator*{\E}{\mathbb{E}}
\DeclareMathOperator*{\dd}{\mathrm{d}}

\begin{document}


\title{Generalizing Alessandro-Beatrice-Bertotti-Montorsi (AMMB) Models to the Case of Velocity-dependent Dissipation}


\author{Wong Ka Sin, Jamie}
\affiliation{Physics Department, Chinese University of Hong Kong}


\date{\today}

\begin{abstract}
We study a more general class of the Alessandro-Beatrice-Bertotti-Montorsi (AMMB) models with velocity-dependent dissipation. We obtain the Fokker-Planck equation describing the evolution of an arbitrary initial probability distribution, and from there the stationary distribution under constant driving. For this class of models, we show that the distribution of the size of an avalanche is the same as when the dissipation is velocity-independent. As for durations, we show that, under non-stationary driving known as "kicks", although no closed-form solution seems to be available for an arbitrary velocity-dependent dissipation, for large durations the distribution seems to demonstrate an exponential fall-off, while for small durations (under some extra conditions to be made clear in the paper) the distribution seems to show a characteristic power-law behaviour.
\end{abstract}


\maketitle

\section{Introduction}

The study of how elastic interfaces are driven through a disordered environment is a fascinating topic. Due to the presence of pinning interactions, the motion is haphazard and proceeds in random jumps known as \textit{avalanches}. The basic quantities of interest are the \textit{size} (i.e. Once an avalanche begins for how much the interface moves?) and \textit{duration} (for how long the interface moves?) of an avalanche. It is known that both quantities demonstrate characteristic power-laws \cite{colaiori2008exactly}. Of course, interest in the subject would have been less fervent if its application had been confined to the study of elastic interfaces. On the contrary , one finds that the subject is closely related to the study of Barkhausen crackling noise (i.e. haphazard response of soft/hard magnets to external magnetizing fields)\cite{colaiori2008exactly} and earthquakes \cite{fisher1998collective,jagla2014viscoelastic}.

A famous model within the study of the aforementioned dynamic systems is the Alessandro-Beatrice-Bertotti-Montorsi (ABBM) model \cite{alessandro1990domain,colaiori2008exactly}. Its popularity may be attributed a number of factors, not least of which are its success in capturing the phenomenology and its analytical tractability. In \cite{dobrinevski2012nonstationary}, it is shown that it truly describes the motion of the centre of mass of an interface under the assumption of independent Brownian pinning forces. 

Now the ABBM model describes dissipative dynamics (i.e. the equation of motion contains only a first-order time derivative). In the "standard" model, the dissipative coefficient $\eta$ is independent of the velocity of the interface. One sensible question is how the "standard" ABBM model may be generalised to the case of velocity-dependent dissipation, as many physical systems exhibit velocity-dependent dissipative constants. In \cite{dobrinevski2012nonstationary} it was deemed to be an important step towards more realistic earthquake models, as rocks in real life follow state-and-rate friction laws \cite{scholz1998earthquakes}. With this goal in mind, we shall study the following model,
\begin{eqnarray}
\eta\left[\partial_t h(t)\right] \partial_t h(t) &=& F\left(h(t)\right) - k_0 \left[ h(t) - w(t) \right]
,
\label{eq:basic}
\end{eqnarray} 
where $h(t)$ is the height function, $w(t)$ is the external driving, $\eta\left[\partial_t h(t)\right]$ is a velocity-dependent dissipation coefficient, $F(h)$ is an effective random force - sum of the local pinning forces - which is postulated to be a Gaussian process in the variable $h$ with the following correlation structure,
\begin{eqnarray}
\E \left[ F(h_1) - F(h_2) \right]^2 = 2 \sigma \left| h_1 - h_2 \right|
.
\label{eq:define_F_h}
\end{eqnarray}
In the previous expression $\E$ denotes the expectation operator. We impose the extra condition that the function $\Phi(x) := \eta(x) x$ be continuous and monotonically increasing in $x$, with $\
\Phi(0) = 0$. In the following we shall use the notation $\dot{h} := \partial_t h$. Moreover, a stochastic process (e.g. $h(t)$) has a natural time-dependence, and we shall denote it by either $h(t)$ or $h_t$. For example, the value of $h$ at $t=0$ shall be denoted by $h(0)$ or $h_0$.

Our paper shall be organised as follows. In Section~(\ref{sec:basicConsiderations}), we outline some basic properties of our model. In Section~(\ref{sec:FP}), we derive the Fokker-Planck Equation which describes the time-evolution of an arbitrary initial probability distribution of the velocity. In Section~(\ref{sec:Pstat}), we derive the stationary distribution under constant driving, $w(t) = v t$. In Section~(\ref{sec:size}), we argue that the size distribution is the same as when $\eta$ is velocity-independent. In Section~(\ref{sec:duration}), we study the long-time and short-time behaviours of durations under non-stationary driving, $w(t) = w \delta(t)$. We conclude our paper in Section~(\ref{sec:conclusion}). For ease of presentation, we shall organise our paper in a series of claims and proofs. Note that the proofs are more physically-oriented than are mathematically rigorous.

\section{Basic Considerations}
\label{sec:basicConsiderations}

Before proceeding, let us clarify the differentiability of various quantities in Eq.~(\ref{eq:basic}). We claim that:

\begin{claim}
The function $h(t)$, whose dynamics is governed by Eqs~(\ref{eq:basic}) and (\ref{eq:define_F_h}), is continuous and differentiable in $t$, while $\dot{h}(t)$ is continuous but not differentiable. 
\end{claim}

\begin{proof}
Assume that $h(t)$ has a step-like discontinuity at $t=t_0$. Then, $F(h(t))$, as a brownian motion in $h$, shall have a step-like discontinuity at $h(t_0)$, and the R.H.S. of Eq.~(\ref{eq:basic}) shall have at most a step-like discontinuity, but $\dot{h}(t)$ shall contain a delta-function at $t=t_0$, and so will the L.H.S.  of Eq.~(\ref{eq:basic}), which leads to a contradiction. Thus, $h(t)$ is continuous in $t$. If $h(t)$ is continuous, then the R.H.S. of Eq.~(\ref{eq:basic}) will also be continuous, and this implies $\Phi[\dot{h}(t)]$, and thus, $\dot{h}(t)$ (through the monotonicity and continuity of $\Phi(x)$) are continuous. Note however that $\dot{h}(t)$ is non-differentiable, because $F(h)$, as a brownian motion, is non-differentiable in $h$ almost surely.
\end{proof}

So what are the implications of the above claim? It implies, if one is to differentiate Eq.~(\ref{eq:basic}), one can at most differentiate it in Ito's sense \cite{klebaner2005introduction}. Taking Ito's differential of Eq.(\ref{eq:basic}), one obtains,
\begin{eqnarray}
\dd y &=& \dd F(h) - k_0 \left[ \dot{h} - \dot{w}(t) \right] \dd t 
,
\\
y &=& \Phi(\dot{h})
.
\end{eqnarray} 
What is $\dd F(h)$? It turns out that $\dd F(h) = \sqrt{2 \sigma \left| \dot{h} \right|} \dd W_t$, where $W_t$ is a standard brownian motion with $\E\left[ W_t - W_s \right]^2 = \left|t-s\right|$. 

\begin{claim}
The stochastic process described by the Ito's differential $\dd F\left(h(t)\right)$ is equivalent to the process $\sqrt{2 \sigma \dot{h}} \dd W_t$, where $W_t$ is the standard Brownian motion with $\E\left[ W_t - W_s \right]^2 = \left| t-s \right|$ 
\end{claim}

\begin{proof}
To see this, consider the infinitesimal increments of both processes. Then $\dd F(h(s))$ ($= F\left[ h(s+\delta s) \right] - F\left[h(s)\right]$) is normally-distributed with mean $0$ and second moment $\E\left[ \dd F(h(s)) \right]^2 = 2 \sigma \dot{h}$ (by the differentiability of $h(t)$), while $\sqrt{2 \sigma \left| \dot{h} \right|} \dd W_t$ is normally-distributed with the same mean and same second moment. Moreover, both processes have independent increments.
\end{proof}

Strictly speaking, to make the previous identification, one has to ensure that $\dot{h} \geq 0$ throughout. Within the framework of ABBM model, this is known as Middleton's theorem \cite{middleton1992asymptotic}, which roughly states that for monotonous driving, $\dot{w}(t) \geq 0$, $\dot{h}(t) \geq 0$. 

\begin{claim}[Middleton's Theorem]
For monotonous driving, $\dot{w}(t) \geq 0$, if $\dot{h}(0) \geq 0$, then $\dot{h}(t) \geq 0$. Moreover, suppose that the system starts off with a non-zero velocity $\dot{h}(t=0) > 0$ and $\dot{w}(t) = 0$ for $t > 0$. If $\dot{h}(t_0) = 0$, $\dot{h}(t) = 0$ for $t \geq t_0$.
\end{claim}

\begin{proof}
To show that Middleton's theorem holds within our model, write 
\begin{eqnarray}
\dd y &=& \sqrt{2 \sigma \dot{h} } \dd W_t - k_0 \left[ \dot{h} - \dot{w}(t) \right] \dd t 
,
\label{eq:+ve_h_dot}
\\
y &=& \Phi(\dot{h})
\label{eq:define_y}
,
\end{eqnarray} 
which holds supposedly when $\dot{h} \geq 0$. Now assume that we start off with a positive $\dot{h}$ (i.e. $\dot{h}(t=0) > 0$). If $h(t)$ turns negative at a later time, then there must be a time $t_c$ at which $\dot{h}(t_c) = 0$ (due to the continuity of $\dot{h}(t)$). At this $t_c$, Eq.~(\ref{eq:+ve_h_dot}) presumably holds. The first term on the R.H.S. vanishes while the second term is non-negative for $\dot{w}(t) \geq 0$. Therefore, $y$, and thus $\dot{h}$ (through the continuity and monotonicity of $\Phi(\dot{h})$), will be almost surely non-negative immediately after it hits $\dot{h} = 0$, and a negative value of $\dot{h}$ can never be reached. Moreover, if $\dot{w}(t) = 0 $ for $t > 0$, the R.H.S. will be strictly zero, and $\dot{h}(t) \geq 0$ for $t \geq t_c$.
\end{proof}

As a result of Middleton's theorem, we assert that Eq.~(\ref{eq:+ve_h_dot}) always holds for monotonous driving.

\section{Derivation of the Fokker-Planck Equation}
\label{sec:FP}

For those who are not familiar with the machinery of Ito's Calculus, let us state a mnemonic, $\left( \dd W_t \right)^2 = \dd t$, which may be interpreted naively as saying quantities of the order $\left( \dd W_t\right)^2$, which is of the same order as $\dd t$, \textit{cannot} be ignored. Thus, for an arbitrary function $f(W_t,t)$, 
\begin{eqnarray}
\dd f(W_t,t) &=& \partial_t f(W_t,t) \dd t + \partial_{W_t} f(W_t,t) \dd W_t 
\nonumber\\
&&+ \frac{1}{2} \partial_{W_t}^2 f(W_t,t) \dd t
.
\end{eqnarray}
The last term on the R.H.S. of the previous expression is peculiar to a function depending on $W_t$.

Now consider an arbitrary function $\Psi(\dot{h}, t)$. We want to describe the time-evolution of $\Psi(y, t)$ \footnote{Note that $\dot{h} = \Phi^{-1}(y)$}, which is governed by the following Ito's differential,
\begin{eqnarray}
\dd \Psi(y, t) 
&=&
\partial_t \Psi(y, t) \dd t + \partial_{y} \Psi(y,t) \dd y + \frac{1}{2} \partial_{y}^2 \Psi(y,t) \left( \dd y \right)^2
\nonumber\\
&=&
\left\{ \partial_t \Psi(y, t) - k_0 \left[ \dot{h} - \dot{w}(t) \right] \partial_{y} \Psi(y,t) \right\} \dd t 
\nonumber\\
&&+  \sqrt{2 \sigma \dot{h}} \partial_{y} \Psi(y,t) \dd W_t  
\nonumber\\
&&+ \frac{1}{2} \partial_{y}^2 \Psi(y,t) \left( 2 \sigma \dot{h} \right) \dd t
.
\end{eqnarray}
The Feymann-Kac equation, which describes the evolution of the conditional expectation $\Psi(y,t) = \E \left[ f(y_T) | y_t=y \right]$ from $t=t_0$ to $t=T$, where $f(y_T)$ is an arbitrary function of $y_T$, is obtained by setting the portion of $\dd \Psi(y, t)$ proportional to $\dd t$ to $0$,
\begin{equation}
-k_0 \left[ \dot{h} - \dot{w}(t) \right] \partial_y \Psi(y,t) + \sigma \dot{h} \partial^2 \Psi(y,t) + \partial_t \Psi(y,t) 
= 0
.
\end{equation}
The Feymann-Kac equation should be solved with the terminal boundary condition $\Psi(y,T) = f(y)$ (because $ \Psi(y,T) = \E \left[ f(y_T) | y_T=y \right] = f(y)$). We summarize the above in the following claim:

\begin{claim}[Feymann-Kac Equation]
Suppose $t \leq T$. The Feymann-Kac equation, which describes the evolution of the conditional expectation $\Psi(y, t) := \E \left[ f(y_T) | y_t = y \right]$, is given by
\begin{equation}
-k_0 \left[ \dot{h} - \dot{w}(t) \right] \partial_y \Psi(y,t) + \sigma \dot{h} \partial^2 \Psi(y,t) + \partial_t \Psi(y,t) 
= 0
.
\label{eq:FK}
\end{equation}
The previous equation should be solved with the terminal boundary condition $\Psi(y,T) = f(y)$.
\end{claim}

The Fokker-Planck equation, on the other hand, describes the evolution of the conditional probability distribution $P(y_t=y | y_0)$ (i.e. the probability distribution of $y_t$ given the value of $y_0$). It is easy to obtain it given one knows the form of Feymann-Kac equation.

\begin{claim}[Fokker-Planck Equation]
The Fokker-Planck equation, which describes the forward evolution of the conditional probability distribution $P(y,t) := \textrm{Prob.}(y_t=y | y_0)$, is given by
\begin{equation}
-k_0 \partial_y \left[ \left( \dot{h} - \dot{w}(t)\right) P(y,t) \right]
-\sigma \partial_y^2 \left[ \dot{h} P(y,t) \right] + \partial_t P(y,t) 
= 0
.
\label{eq:FP}
\end{equation}
It should be solved with the initial boundary condition $P(y,0) = P_0(y)$, where $P_0(y)$ is the (initial) distribution function of $y_0$.
\end{claim}

\begin{proof}
Noting that $\int \dd y \Psi(y, t) P(y_t=y | y_0) =\int \dd y \E \left[ f(y_T) | y_t=y \right] P(y_t=y | y_0) = \E \left[ f(y_T) | y_0 \right]$, which is independent of $t$, one obtains the relation $\partial_t \int \dd y \Psi(y, t) P(y_t=y | y_0) = 0$. Interchanging the order of differentiation and integration, using Eq.~(\ref{eq:FK}) to replace $\partial_t \Psi(y, t)$, and integrating by parts, one arrives at Eq.~(\ref{eq:FP})
\end{proof}

In the following we shall concentrate on the study of Eq.(\ref{eq:FK}).

\section{Stationary Distribution}
\label{sec:Pstat}

\begin{claim}
The stationary distribution to Eq.(\ref{eq:FK}) under constant driving $\dot{w}(t) = v$ is given by, $P_{\textrm{stat}}(y) = N \exp\left\{ -\frac{k_0}{\sigma} \left[ y - v \int \dd y/\Phi^{-1}(y) \right] \right\} / \dot{h}$, where $N$ is a normalization constant.
\end{claim}

\begin{proof}
Since we are dealing with stationary distributions, $\partial_t P_{\textrm{stat}}(y) = 0$. Eq.~(\ref{eq:FK}) becomes
\begin{equation}
- k_0 \partial_y \left[ (\dot{h}-v)P_{\textrm{stat}}(y) \right] - \sigma \partial_y^2 \left[ \dot{h} P_{\textrm{stat}}(y) \right] = 0
. 
\end{equation}
The above equation may be interpreted as $\partial_y J(y) = 0$, where $J(y)$ is the probability current. If we do not want a probability current to flow through $y=0$ (i.e. $J(0)=0$), we may simply integrate the above equation with respect to $y$, without introducing an extra integration constant. The stationary distribution follows from simple integrating factor methods. 
\end{proof}

Of course, the expression for $P_{\textrm{stat}}(y)$ holds only when $\int \dd y/\Phi^{-1}(y)$ exists. 

\section{Size Distribution}
\label{sec:size}

The size $S$ of an avalanche is understood as follows. Suppose at $t=0$, we impart a positive velocity $\dot{h}(t=0)$ to a particle. Then, what is the size of $h(t_s)$, where $t_s$ is the first time $\dot{h}$ reaches $0$? Let us go back to Eq.~(\ref{eq:basic}). The L.H.S. equals $\Phi(\dot{h})$, and $\Phi(\dot{h})=0$ iff $\dot{h}=0$. Thus, $h(t_s)$ satisfies $0=F\left[ h(t_s) \right] - k_0 \left[ h(t_s) - w(t_s) \right]$, or, if one sets $S = h(t_s)$, $0=F\left[ S \right] - k_0 \left[ S - w(t_s) \right]$, but the latter is exactly the same equation the random variable $S$ satisfies when $\eta$ is velocity-independent \cite{dobrinevski2012nonstationary}.

\begin{claim}
The distribution of the size of an avalanche $S$ follows the same distribution as when the dissipation coefficient $\eta$ is velocity-independent.
\end{claim}

\section{Duration Distribution Under Nonstationary Driving ("Kicks")}
\label{sec:duration}

The duration $T$ of an avalanche event is understood as follows. Suppose at $t=0$, we impart a positive velocity $\dot{h}(t=0)$ to a particle. Then, what is $t_s$ ($=T$), the time when $\dot{h}(t_s)$ first hits $0$? It turns out that no obvious, closed-form solutions exist for general forms of $\Phi(\dot{h})$. So instead let us investigate whether some common properties exist among the different admissible forms of $\Phi(\dot{h})$. For simplicity we shall exert non-stationary driving, referred to as "a kick" \cite{dobrinevski2012nonstationary},  (i.e. $w(t) = w \delta(t)$). Let us start our investigation with the long-time behaviour of $T$.

We need the following fact for the investigation of duration distribution \cite{colaiori2008exactly}

\begin{claim}[Calculation of Duration Distribution]
The probability distribution for $T$, $P_T(T=t)$, is obtained as follows. Solve Eq.~(\ref{eq:FP}), subjected to $P(y,0) = \delta (y - y_0)$. $P_T(T=t)$ is obtained as the value of the probability current (which should be flowing towards $y=0$) at $y=0$. (i.e. $P_T(T=t) = - k_0  (\dot{h}-\dot{w}(t))P(y=0, t)  - \sigma \partial_y \left[ \dot{h} P(y=0, t) \right]$) The "kick" $w$ is related to the initial value of $\dot{h}$ through $\Phi\left[\dot{h}(0)\right] = k_0 w$.
\label{claim:duration}
\end{claim}

\subsection{Long-time Behaviour}

For $w(t)=0$, Eq.~(\ref{eq:FP}) becomes,
\begin{equation}
- k_0 \partial_y \left[ \dot{h} P(y,t) \right] - \sigma \partial_y^2 \left[ \dot{h} P(y,t) \right] + \partial_t P(y,t) = 0
.
\label{eq:FP_zeroDriving}
\end{equation}
Defining the function $G(y,t) = \dot{h} P(y,t)$, performing separation of variables, $G(y,t) = \sum_{\lambda} e^{-\lambda^2 t} r_{\lambda}(y)$, and noting that $\dot{h} = \Phi^{-1}(y)$, one obtains the characteristic O.D.E. for $\lambda$,
\begin{equation}
-k_0 \partial_y r_{\lambda}(y) - \sigma \partial_y^2 r_{\lambda}(y) - \frac{\lambda^2}{\Phi^{-1}(y)} r_{\lambda}(y) = 0
.
\end{equation}
Doing the final transformation $g_{\lambda}(y) = e^{k_0 y /2 \sigma} r_{\lambda}(y)$, one obtains
\begin{equation}
-\sigma \partial_y^2 g_{\lambda}(y) + \sigma \left( \frac{k_0}{2 \sigma} \right)^2 g_{\lambda}(y) = \frac{\lambda^2}{\Phi^{-1}(y)} g_{\lambda}(y)
.
\label{eq:char_ODE}
\end{equation}
The final form of the characteristic O.D.E., Eq~(\ref{eq:char_ODE}), is quite nice and very "amenable" to analysis.

\begin{claim}[Asymptotics of solutions to the Characteristic O.D.E.]
For given $\lambda$, for $y$ sufficiently large, there exist two solutions to Eq.~(\ref{eq:char_ODE}), one monotonically increasing and the other monotonically decreasing.
\end{claim}

\begin{proof}
Since $\Phi(\dot{h})$ is monotonically increasing, so is $\Phi^{-1}(y)$ and $1/\Phi^{-1}(y)$ is monotonically decreasing. Thus for large enough $y > y_0$, $\sigma \left( \frac{k_0}{2 \sigma} \right)^2 g_{\lambda}(y) \gg \frac{\lambda^2}{\Phi^{-1}(y)} g_{\lambda}(y)$, and 
\begin{equation}
\sigma \partial_y^2 g(y) = \left\{  \sigma \left( \frac{k_0}{2 \sigma} \right)^2 - \frac{\lambda^2}{\Phi^{-1}(y)} \right\} g_{\lambda}(y) > 0
.
\end{equation}
Consider two solutions $g_{+}(y)$, $g_{-}(y)$ with $( g_{+}(y_0),g^{'}_{+}(y_0) )  = (1, 1)$, and $( g_{-}(y_0),g^{'}_{-}(y_0) )  = (1, -1)$. Then, since $\partial_y^2 g_{+}(y) > 0$, $g_{+}(y)$ stays increasing while $g_{-}(y)$ stays increasing. Indeed when $1/\Phi^{-1}(y)$ decays sufficiently fast, Eq.~(\ref{eq:char_ODE}) becomes
\begin{equation}
\partial_y^2 g(y) \approx  \left( \frac{k_0}{2 \sigma} \right)^2 g_{\lambda}(y)
,
\end{equation}
and one may claim futher that $g_{+}(y) \sim \exp\left( \frac{k_0}{2 \sigma} y \right)$ and $g_{-}(y) \ \exp\left( -\frac{k_0}{2 \sigma} y \right)$.

\end{proof}

Since $P(y,t) = G(y,t)/\dot{h} = \left( 1/\dot{h} \right)\sum_{\lambda} e^{-\lambda^2 t} r_{\lambda}(y) = \left( 1/\dot{h} \right)\sum_{\lambda} e^{-\lambda^2 t} e^{-k_0 y /2 \sigma} g_{\lambda}(y)$, and one expects the probability distribution $P(y,t)$ to be integrable, it seems that

\begin{claim}
In solving Eq.~(\ref{eq:char_ODE}), one should impose the boundary conditions that $g_{\lambda}(y) / \dot{h} $ be regular at $y=0$ and $g_{\lambda}(y)$ be exponentially decaying for large $y$.
\end{claim}

For example, for $\Phi(\dot{h})$ such that $\Phi^{-1}(y) \sim y$ for small $y$, one may obtain a Frobenius  series solution for Eq.~(\ref{eq:char_ODE}), and the correct $g_{\lambda}(y) \sim y$. For arbitrary $\lambda$, both boundary conditions cannot be satisfied simultaneously, and one is led to the fact,

\begin{claim}
The eigenvalues $\lambda^2$ to Eq.~(\ref{eq:char_ODE}) are discrete.
\end{claim}

Moreover, 

\begin{claim}[Stability of Evolution]
The eigenvalues $\lambda^2$ are all positive.
\end{claim}

\begin{proof}
Multiplying Eq.~(\ref{eq:char_ODE}) by $g(y)$, integrating both sides with $\int_0^{\infty} \dd y$, and using the boundary conditions, one obtains
\begin{eqnarray}
&&\sigma \int_0^{\infty} \dd y \left[ \partial_y g_{\lambda}(y) \right]^2
+ \sigma \left( \frac{k_0}{2 \sigma} \right)^2 \int_0^{\infty} \dd y \left[ g(y) \right]^2
\nonumber\\
&&=
\lambda^2 \int_0^{\infty} \dd y \frac{\left[g(y)\right]^2}{\Phi^{-1}(y)}
.
\end{eqnarray}
\end{proof}

The positivity of $\lambda^2$ ensures that the probability distribution $P(y,t)$ will not blow up as it evolves.

So what are the implications of the above claims? Since $P(y,t) = \left( 1/\dot{h} \right)\sum_{n=0}^{\infty} e^{-\lambda_n^2 t} e^{-k_0 y /2 \sigma} g_{\lambda}(y)$, for $t$ such that $(\lambda_1^2 - \lambda_0^2) t$ is sufficiently large, $P(y,t) \sim e^{-\lambda_0^2 t} e^{-k_0 y /2 \sigma} g_{\lambda}(y) / \dot{h}$. Using Claim~(\ref{claim:duration}), one finds that indeed $P_T(T=t) = - \sigma \partial_{\dot{h}} \left[ \dot{h} P(y,t) \right] = - \sigma \partial_{\dot{h}} \left[ G(y,t) \right]$, and one concludes

\begin{claim}
Define the time scale $\tau = \left( \lambda_1^2 - \lambda_0^2 \right)^{-1}$. If $t/\tau \gg 1$, then $P_T(T=t) \sim e^{-\lambda_0^2 t}$.
\label{claim:largeT}
\end{claim}
 
The timescale $\tau$ is independent of the size of the "kick" $w$. If $w$ is large, $T$ is naturally long, and the exponential behaviour should be more easily observable. Moreover, rigorously speaking, for our arguments preceding Claim~(\ref{claim:largeT}) to be true, we need to show that the spacing between successive eigenvalues is at least finite as $n \rightarrow +\infty$. For the latter we turn to the following non-rigorous claim,

\begin{claim}
For large $n$, $\lambda_n \rightarrow \frac{n \pi}{\delta}$, where $\delta = \int_0^{\infty} \dd y \left[\sigma \Phi^{-1}(y)\right]^{-1/2}$.
\label{claim:eigenvalueLargeN}
\end{claim}

\begin{proof}
For a standard Sturm-Liouville on a finite interval $[a,b]$,
\begin{equation}
\left[ p(x) f^{'}(x)\right]^{'} + q(x) f(x) = \lambda^2 r(x) f(x)
,
\end{equation} 
where $p(x)$ and $r(x)$ do not change sign, the eigenvalues behave asymptotically as, $\lambda_n \rightarrow n \pi / \delta$, where $\delta = \int_a^b\dd x \left[ r(x)/p(x) \right]^{1/2}$ \cite{binding2001existence}. 
\end{proof}

The implications of Claim~(\ref{claim:eigenvalueLargeN})are as follows. If the form of $\Phi^{-1}(y)$ is such that $\int_0^{\infty} \dd y \left[\sigma \Phi^{-1}(y)\right]^{-1/2}$ does not exist, we expect that it is possible for eigenvalues of large $n$ to be "crowded together" (i.e. that $\lambda_n^2$ does not scale as $n^2$), and our argument for the large-time behaviour of durations may be in doubt. On the other hand, to search for numerical evidence in support of our claims, Claim~(\ref{claim:eigenvalueLargeN}) suggests us to adopt a form of $\Phi^{-1}(y)$ such that $\int_0^{\infty} \dd y \left[\sigma \Phi^{-1}(y)\right]^{-1/2}$ exists, with large drive $w$ (which implies that the duration distribution is biased towards the large-time side). 

\subsubsection{Numerical Examples}

For the purpose of finding numerical support for our previous claims, we consider two systems

\begin{enumerate}
\item
System L1: $\Phi(\dot{h}) = \eta \dot{h} - \alpha \dot{h}^2$, $\Phi^{-1}(y) = \left[-\eta + \left( \eta^2 - 4 \alpha y \right)^{1/2} \right]/2 |\alpha|$, $\alpha = -0.25$, $\eta = 1.0$, $k_0 = 1.0$, $\sigma = 1.0$, under driving $w=100.0$
\item
System L2: $\Phi(\dot{h}) = \log\left( 1 + \dot{h} \right)/a$, $\Phi^{-1}(y) = \exp\left(a y\right)-1$, $a = 1.0$, $\eta = 1.0$, $k_0 = 1.0$, $\sigma = 1.0$, under driving $w=100.0$
\end{enumerate}

\begin{figure}
\centering
\subfigure[C.D.F.]{\includegraphics[scale=0.3]{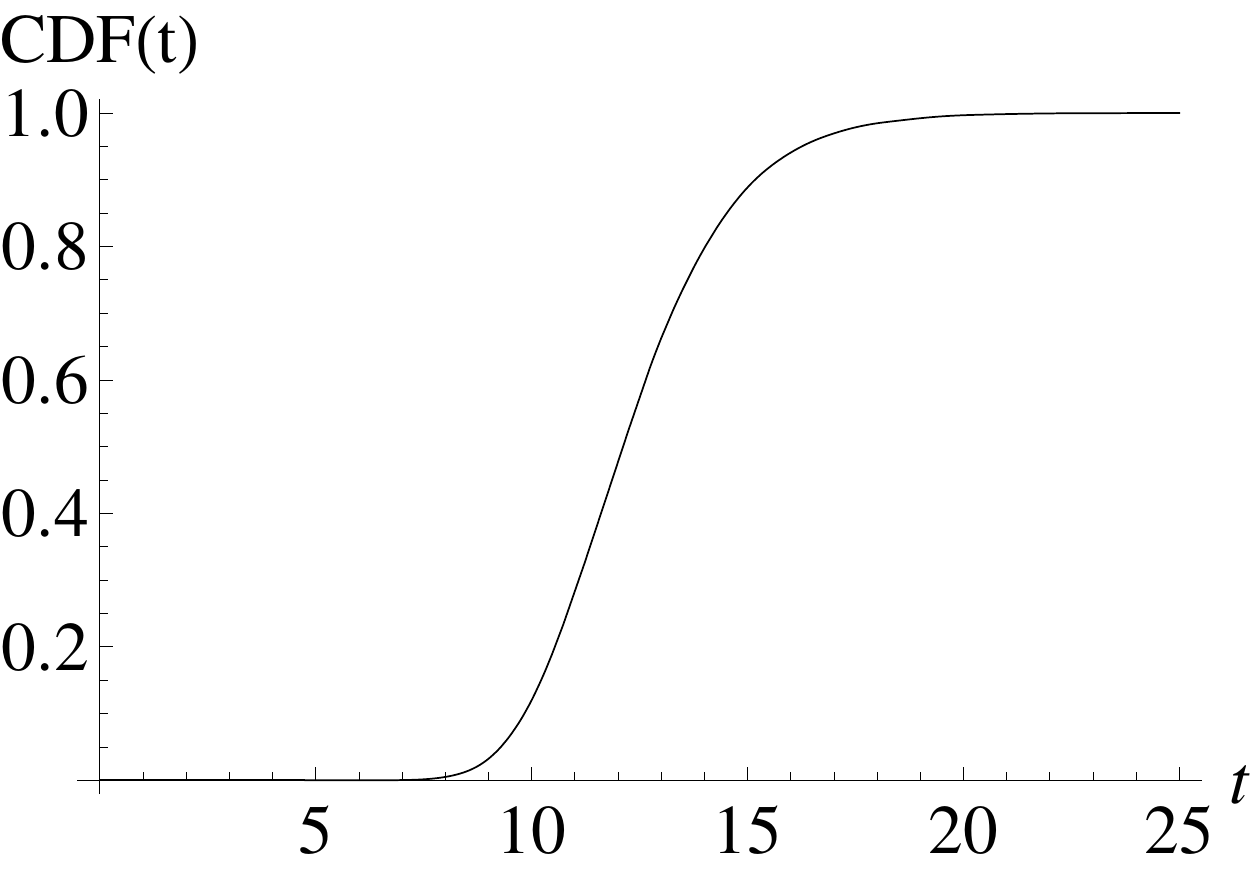}\label{fig:plotCDFquadratic}}\qquad
\subfigure[P.D.F.]{\includegraphics[scale=0.3]{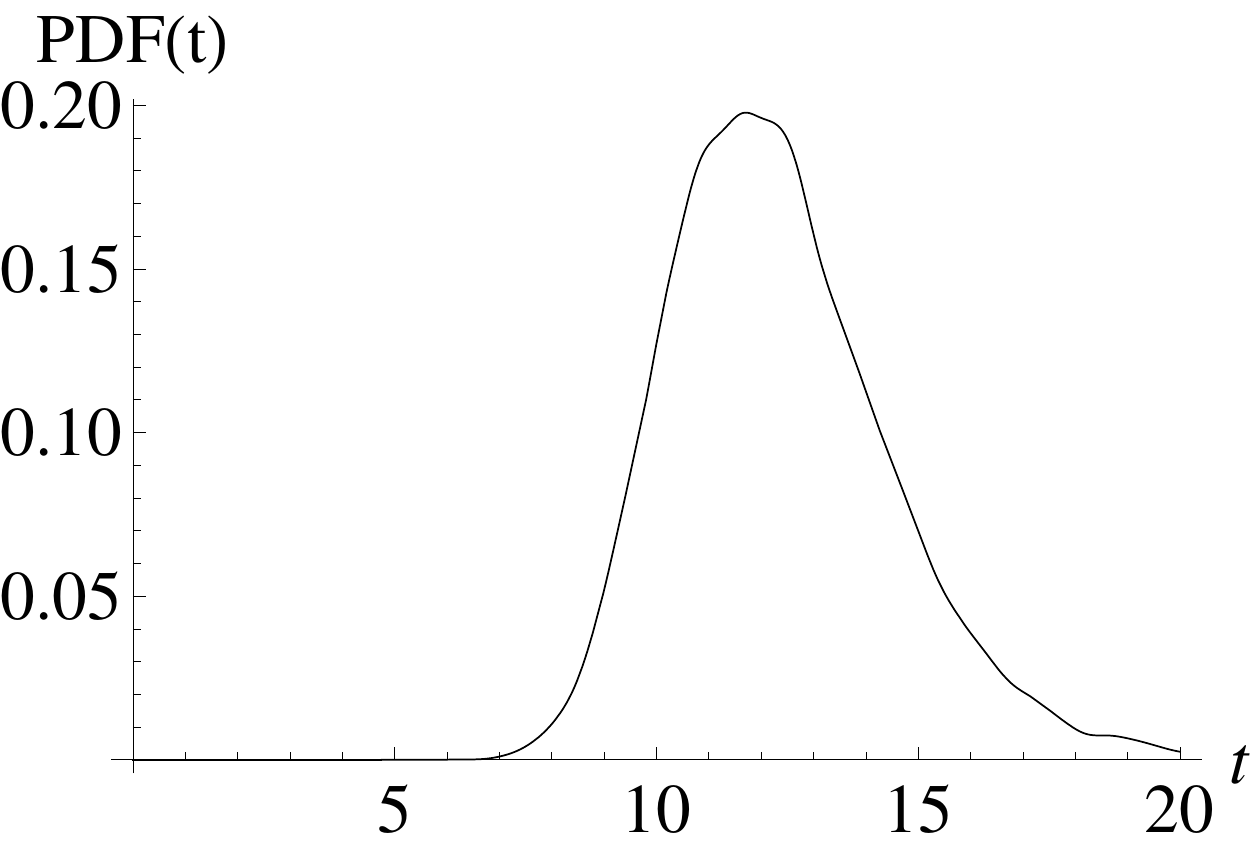}\label{fig:plotPDFquadratic}}\\
\subfigure[Log Plot of P.D.F.]{\includegraphics[scale=0.6]{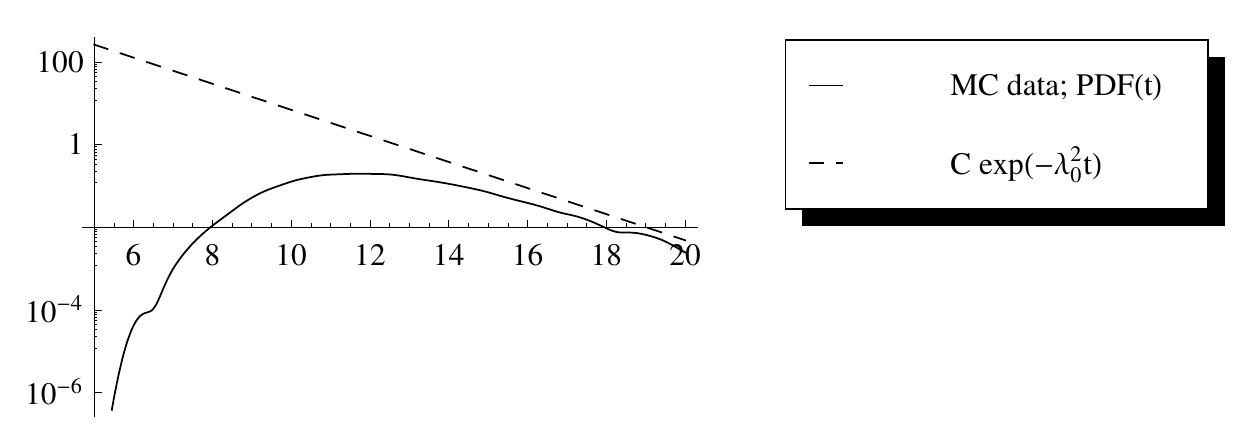}\label{fig:logPlotPDFquadratic}}%
\caption{Fig.~\ref{fig:plotCDFquadratic} shows the cumulative distribution function of duration $T$ for System L1. 20000 paths are sampled in the Monte-Carlo study. The figure shows that the support of the probability density function is roughly within the interval $[5,20]$. Fig.~\ref{fig:plotPDFquadratic} shows the corresponding probability distribution function for System L1. Fig.~\ref{fig:logPlotPDFquadratic} gives the log plot of the P.D.F. of duration $T$ for System L1. Shown also is the function $C \exp(-\lambda_0^2 t)$, where $\lambda_0$ is the smallest eigenvalue for the characteristic O.D.E. for System L1, for comparison}
\label{fig:system1}
\end{figure}

\begin{figure}
\centering
\subfigure[C.D.F.]{\includegraphics[scale=0.3]{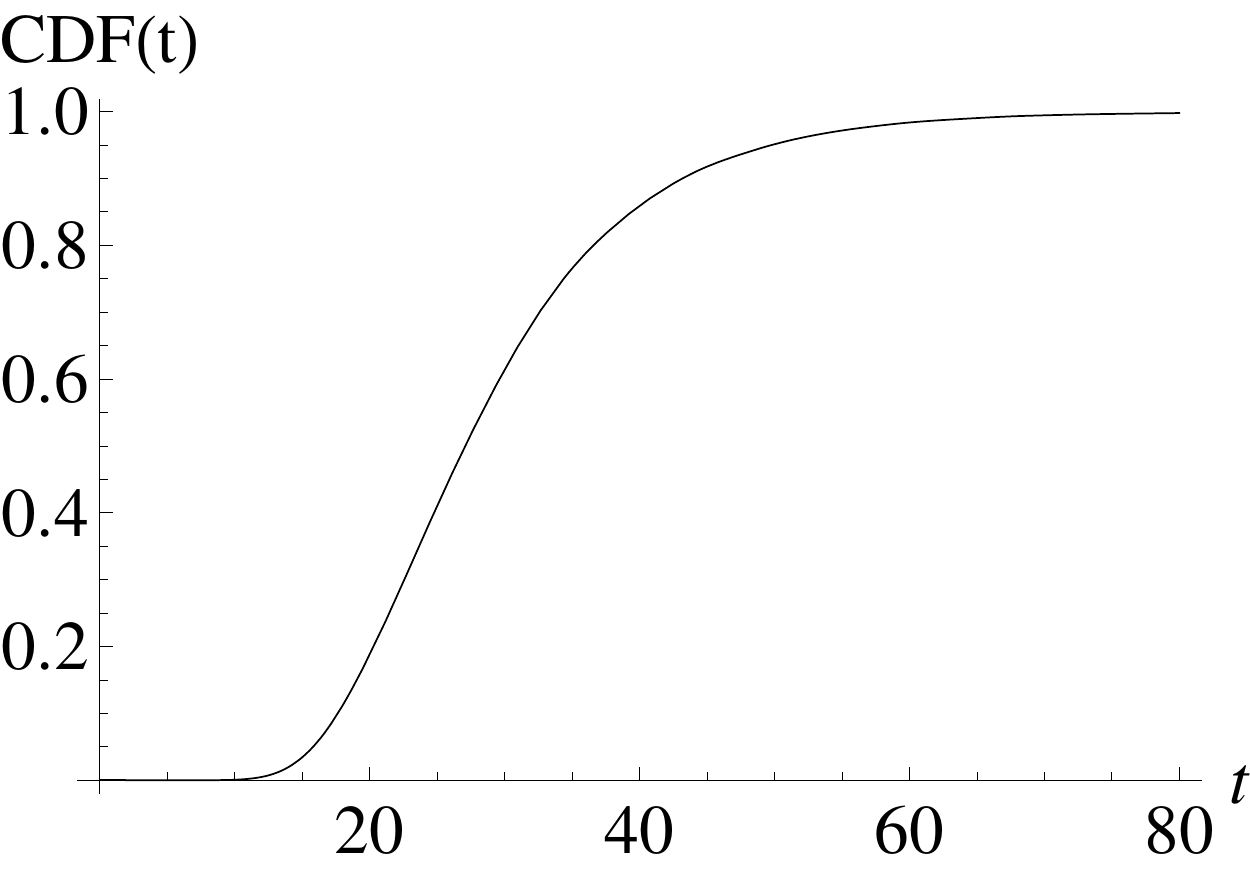}\label{fig:plotCDFexpModel}}\qquad
\subfigure[P.D.F.]{\includegraphics[scale=0.3]{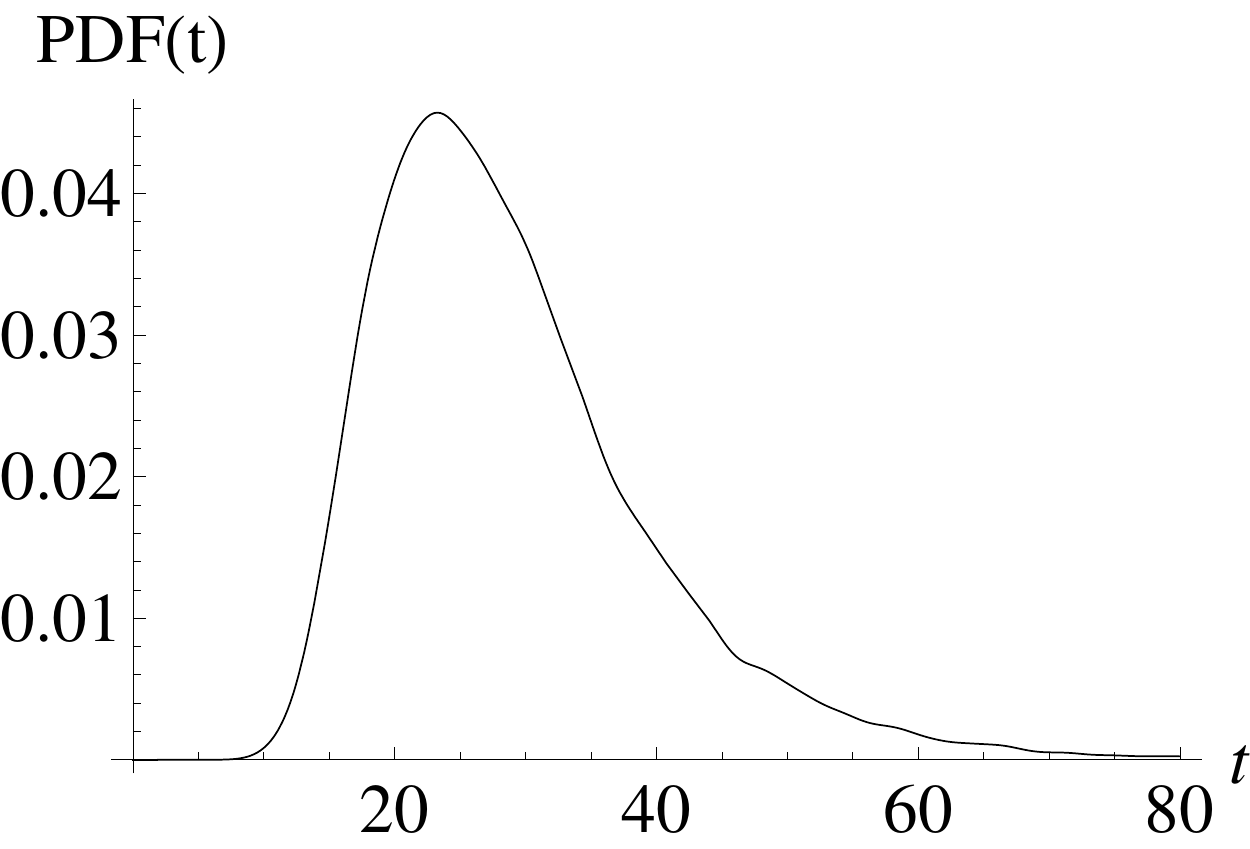}\label{fig:plotPDFexpModel}}\\
\subfigure[Log Plot of P.D.F.]{\includegraphics[scale=0.6]{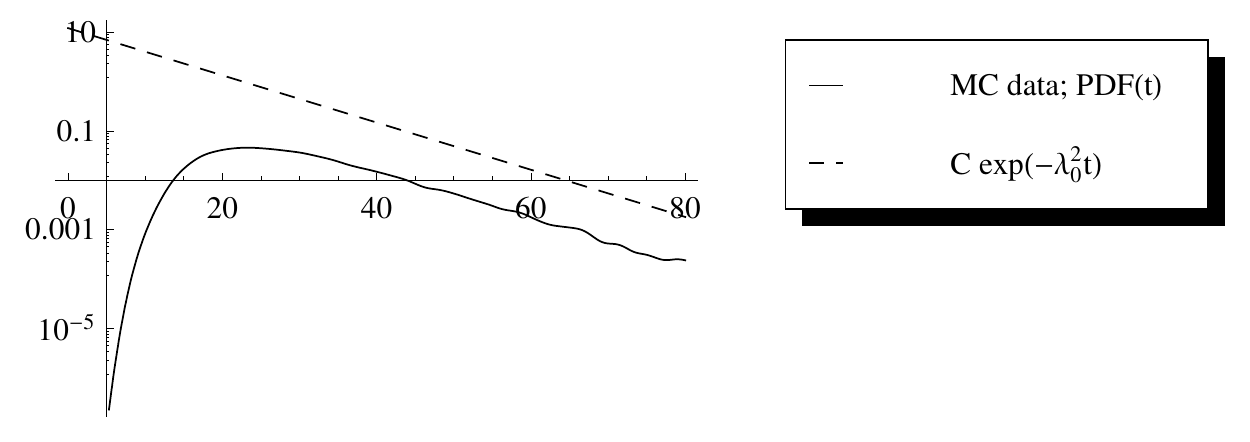}\label{fig:logPlotPDFexpModel}}%
\caption{Fig.~\ref{fig:plotCDFexpModel} shows the cumulative distribution function of duration $T$ for System L2. 20000 paths are sampled in the Monte-Carlo study. The figure shows that the support of the probability density function is roughly within the interval $[10,80]$. Fig.~\ref{fig:plotPDFexpModel} shows the corresponding probability distribution function for System L2. Fig.~\ref{fig:logPlotPDFexpModel} gives the log plot of the P.D.F. of duration $T$ for System L2. Shown also is the function $C \exp(-\lambda_0^2 t)$, where $\lambda_0$ is the smallest eigenvalue for the characteristic O.D.E. for System L2, for comparison}
\label{fig:system1}
\end{figure}

For System L1, the two smallest eigenvalues to Eq.~(\ref{eq:char_ODE}) are $\lambda_0^2 = 0.727$ and $\lambda_1^2 = 1.23$ respectively. Note that for System L1, $\int_0^{\infty} \dd y \left[\sigma \Phi^{-1}(y)\right]^{-1/2}$ does not exist. The characteristic timescale $\tau = 1/(\lambda_1^2 - \lambda_0^2) = 1.99$, and we expect the exponential behaviour to set in for $T \gg \tau = 1.99$. Figures~(\ref{fig:plotCDFquadratic}) and ~(\ref{fig:plotPDFquadratic}) show the cumulative distribution function (C.D.F.) and probability density function (P.D.F.) of durations from a Monte-Carlo simulation for System L1. The C.D.F. shows that the support of the P.D.F. is roughly within the interval $[5,20]$, and we thus expect the exponential behaviour to be observable (although the asymptotic behaviour of larger eigenvalues is unknown). Figure~(\ref{fig:logPlotPDFquadratic}) shows the log-plot of the P.D.F. for System L1, and the function $C e^{-\lambda_0^2 t}$. It is seen that the large-$t$ behaviour is reasonably well (though not perfectly) accounted for by the fit.

For System L2, the two smallest to Eq.~(\ref{eq:char_ODE}) are $\lambda_0^2 = 0.11$ and $\lambda_1^2 = 0.24$ respectively. Note that for this system $\int_0^{\infty} \dd y \left[\sigma \Phi^{-1}(y)\right]^{-1/2}$ exists. The characteristic timescale $\tau = 1/(\lambda_1^2 - \lambda_0^2) = 7.69$, and we expect the exponential behaviour to set in for $T \gg \tau = 7.69$. Figures~(\ref{fig:plotCDFexpModel}) and ~(\ref{fig:plotPDFexpModel}) show the cumulative distribution function (C.D.F.) and probability density function (P.D.F.) of durations from a Monte-Carlo simulation for System L1. The C.D.F. shows that the support of the P.D.F. is roughly within the interval $[10,80]$, and we thus expect the exponential behaviour to be observable. Figure~(\ref{fig:logPlotPDFquadratic}) shows the log-plot of the P.D.F. for System L1, and the function $C e^{-\lambda_0^2 t}$. It is again seen that the large-$t$ behaviour is quite well accounted for by the fit.

\subsection{Short-time Behaviour for Small "Kicks"}

To investigate the short-time behaviour of $P(y,t)$, let us recast the Fokker-Planck Equation for zero driving into a heat conduction equation on a finite interval.

\begin{claim}
The function $G(y,t) := \dot{h} P(t,y) = \Phi^{-1}(y) P(t,y)$ satisfies the following heat conduction equation,
\begin{equation}
\partial_t G(z,t) = \rho(z) \partial_z^2 G(z,t)
,
\end{equation}
where $\rho(z) = \sigma \left( 1 - \frac{k_0}{\sigma} z \right)^2 \Phi^{-1}(z)$ and $z = \frac{\sigma}{k_0} \left[ 1 - \exp\left( - \frac{k_0}{\sigma} y \right) \right]$.
\label{claim:heatConduction}
\end{claim}

\begin{proof}
Start from Eq.~(\ref{eq:FP_zeroDriving}) and perform the stated change of variable $z = \frac{\sigma}{k_0} \left[ 1 - \exp\left( - \frac{k_0}{\sigma} y \right) \right]$.
\end{proof}

The above equation encourages us to adopt the following views. If $w$ is small, then $\dot{h}(0) = \Phi^{-1}(k_0 w) $ will also be small, and so are $y(t=0) := y_0$ and $z(t=0) := z_0$. For $z \sim 0$, one may adopt the approximation $\rho(z) \approx C \sigma z^{\alpha}$, where $\Phi^{-1}(z) \approx C z^{\alpha}$ for small $z$. The range of validity of the approximation is as follows. Suppose $z_{\textrm{min}}$ is the first smallest $z$ for which $\rho_{\textrm{approx}}(z)$ deviates significantly from $\rho(z)$. We solve $G_{\textrm{approx}}(z,t)$ using $\rho_{\textrm{approx}}(z)$, subjected to the boundary condition $G_{\textrm{approx}}(z,0) = \Phi^{-1}(z) \delta(z-z_0)$. Then the approximation is valid for $t$ such that $G_{\textrm{approx}}(z_{\textrm{min}},t) \approx 0$. (i.e. A particle starting at $z=z_0$ has negligible chance of "diffusing" to $z_{\textrm{min}}$.)

From a physical standpoint, we are interested in cases where $\Phi(\dot{h}) \sim \dot{h}^N$, in which $N$ is an integer. For $N=1$, we retrieve the case for a velocity-independent dissipation coefficient. If $\Phi(\dot{h}) \sim \dot{h}^N$, then $\Phi^{-1}(y) \sim y^{1/N}$ and also $\Phi^{-1}(z) \sim z^{1/N}$ (from the transformation in Claim~(\ref{claim:heatConduction})). Assuming $\Phi^{-1}(z) \approx (z/C_N)^{1/N}$ for small $z$, one is led to consider the following heat conduction equation,
\begin{equation}
\partial_t G(z,t) = \sigma (z/C_N)^{1/N} \partial_z^2 G(z,t)
.
\end{equation}
Moreover, at first we need to consider a boundary condition for $G(z,t)$ at $z=1$, namely, that $G(z,t)$ be regular at $z=1$, but now this boundary condition seems to be unimportant, for as long as $t$ is such that $G(z_{\textrm{min}},t) \approx 0$, where $z_{\textrm{min}} \ll 1$. 

\begin{claim}
The solution to the equation,
\begin{equation}
\partial_t G(z,t) = \sigma (z/C_N)^{1/N} \partial_z^2 G(z,t),
\end{equation}
subjected to the boundary conditions that 
\begin{enumerate}
\item 
$G(z,t)$ be regular at $z=0$, and that
\item
 $G(z,0) = \Phi^{-1}(z_0) \left( 1-\frac{k_0}{\sigma} z\right)^{-1} \delta(z-z_0)$,
\end{enumerate}
is given by,
\begin{widetext}
\begin{equation}
G(z,t) =
\left( 1 - \frac{1}{2 N} \right) D_N^2 z_0^{\frac{1}{2}-\frac{1}{N}} \left( 1 -\frac{k_0}{\sigma} z_0 \right)^{-1} \Phi^{-1}(z_0) 
\frac{\sqrt{z}}{t} \exp\left\{ -\frac{D_N^2}{t} \left[ z_0^{\frac{2 N -1}{N}} + z^{\frac{2 N -1}{N}} \right] \right\} I_{\frac{N}{2 N -1}} \left( \frac{D_N^2 z_0^{1-\frac{1}{2 N}} z^{1-\frac{1}{2 N}}}{2 t} \right)
,
\end{equation} 
\end{widetext}
where $D_N = \frac{2 N}{2 N - 1} \left( \frac{C_N^{\frac{1}{N}}}{\sigma} \right)^{\frac{1}{2}}$.
\label{claim:GsmallTime}
\end{claim}

\begin{proof}
First let us consider the initial condition for $G(z,0)$. We want the initial probability distribution to satisfy $P(y,0) = \delta( y-y_0 )$. Since $z = \frac{\sigma}{k_0} \left[ 1 - \exp\left( - \frac{k_0}{\sigma} y \right) \right]$, $P(z,0) =  \left( 1-\frac{k_0}{\sigma} z\right)^{-1} \delta( z - z_0 )$. Now as $G(z,t) = \Phi^{-1}(y) P(z,t)$, the initial condition follows.

To obtain the solution, write $G(z,t) = \int \dd \lambda e^{-\lambda^2 t} A(\lambda) r_{\lambda}(z)$. Then $r_{\lambda}(z)$ satisfies the equation,
\begin{equation}
-\lambda^2 r_{\lambda}(z) = \frac{\sigma}{C_N^{1/N}} z^{1/N} \partial_z^2 r_{\lambda} (z)
. 
\end{equation}
The solution which is regular at $z=0$ is given by $r_{\lambda}(z) = \sqrt{z} J_{\frac{N}{2 N - 1}}\left( \frac{N}{2 N - 1} \sqrt{\frac{C_N^{1/N}}{\sigma}} \lambda z^{1-\frac{1}{2 N}} \right)$, where $J_{\nu}(x)$ is the Bessel function of order $\nu$. Thus,
\begin{eqnarray}
G(z,t) &=& \int_0^{\infty} \dd \lambda A(\lambda) e^{-\lambda^2 t}
\nonumber\\
&&\times \sqrt{z} J_{\frac{N}{2 N - 1}}\left( \frac{N}{2 N - 1} \sqrt{\frac{C_N^{1/N}}{\sigma}} \lambda z^{1-\frac{1}{2 N}} \right)
.
\end{eqnarray}
Using the identity
\begin{equation}
\int_0^{\infty} \dd x x J_{\alpha}(x \lambda) J_{\alpha} (x \lambda^{'}) = \delta(\lambda-\lambda^{'})/\lambda
,
\end{equation}
and the initial boundary condition $G(z,0)$, one obtains
\begin{eqnarray}
A(\lambda) &=& \lambda \left( 1-\frac{1}{2 N} \right) D_N^2 z_0^{\frac{1}{2}-\frac{1}{N}} \left( 1 - \frac{k_0}{\sigma} z_0 \right)^{-1} 
\nonumber\\
&& \times J_{\frac{N}{2 N - 1}}\left( \lambda D_N z_0^{1-\frac{1}{2 N}} \right)
.
\end{eqnarray} 
Thus,
\begin{eqnarray}
G(t,z) 
&=&
\int_0^{\infty} \dd \lambda e^{-\lambda^2 t} \lambda 
\nonumber\\
&& \times \lambda \left( 1-\frac{1}{2 N} \right) D_N^2 z_0^{\frac{1}{2}-\frac{1}{N}} \left( 1 - \frac{k_0}{\sigma} z_0 \right)^{-1} 
\nonumber\\
&& \times J_{\frac{N}{2 N - 1}}\left( \lambda D_N z_0^{1-\frac{1}{2 N}} \right)
\nonumber\\
&& \times \sqrt{z} J_{\frac{N}{2 N - 1}}\left( \frac{N}{2 N - 1} \sqrt{\frac{C_N^{1/N}}{\sigma}} \lambda z^{1-\frac{1}{2 N}} \right)
.
\end{eqnarray}
The integration over $\lambda$ in the previous expression may be performed with the following relation (See. Eq.~(48) in \cite{colaiori2008exactly}),
\begin{equation}
\int_0^{\infty} \dd k k J_{\nu} (k \epsilon) J_{\nu} ( k x ) = \frac{1}{t} \exp\left( - \frac{x^2+\epsilon^2}{t} \right) I_{\nu} \left( \frac{\epsilon x}{2 t} \right)
,
\end{equation}
where $I_{\nu}(x)$ is the Neumann function of order $\nu$. The end result is the expression for $G(z,t)$ as stated in the claim.
\end{proof}

So what is the significance of the previous claim? According to Claim~(\ref{claim:duration}), the duration  distribution $P_T(T=t)$ is related to the value of the probability current, $J(\dot{h},t) = -k_0 \dot{h} P(\dot{h},t) - \sigma \partial_y \left[ \dot{h} P(\dot{h},t) \right] = -k_0 G(z,t) - \sigma \partial_z G(z,t)$, at $\dot{h}=0$ (i.e. $z=0$). Thus, the knowledge of $G(z,t)$ enables us to calculate the duration distribution. Since the solution for $G(z,t)$ in Claim~(\ref{claim:GsmallTime}) is expected to be valid for small enough time $t$, we claim that

\begin{claim}
For small enough time $t$, the duration distribution $P_T(T=t)$ scales with $t$ and $z_0$ as $z_0 t^{-1-\frac{N}{2 N-1}} \exp\left[ -\frac{D_N^2}{t} \left( z_0^{\frac{2 N-1}{N}} + z^{\frac{2 N-1}{N}} \right) \right]$, where $z_0 = \frac{\sigma}{k_0} \left[ 1 - \exp\left( - \frac{k_0}{\sigma} y_0 \right) \right]$ and $y_0 = k_0 w$. 
\end{claim}

\begin{proof}
From the expression for $G(z,t)$ in Claim~(\ref{claim:GsmallTime}), and the fact that $I_{\nu}(x) \sim x^{\nu}$ for small $x$, one finds that
\begin{eqnarray}
G(t,z) &\sim& z_0^{\frac{1}{2}-\frac{1}{N}} \Phi^{-1}(z_0) \frac{1}{t} \exp\left\{ -\frac{D_N^2}{t} \left[ z_0^{\frac{2 N -1}{N}} + z^{\frac{2 N -1}{N}} \right] \right\} 
\nonumber\\
&&\times \sqrt{z} \frac{z^{\frac{1}{2}} z_0^{\frac{1}{2}}}{t^{\frac{N}{2 N-1}}}
\nonumber\\
&\sim& z_0 t^{-1-\frac{N}{2 N-1}} z \exp\left\{ -\frac{D_N^2}{t} \left[ z_0^{\frac{2 N -1}{N}} + z^{\frac{2 N -1}{N}} \right] \right\}
.
\end{eqnarray}
From the previous expression for $G(t,z)$, it is easy to calculate that 
\begin{equation}
J(t,z=0) \sim z_0 t^{-1-\frac{N}{2 N-1}} \exp\left\{ -\frac{D_N^2}{t} \left[ z_0^{\frac{2 N -1}{N}} + z^{\frac{2 N -1}{N}} \right] \right\},
\end{equation}
and the result follows from Claim~(\ref{claim:duration}).
\end{proof}

We note that it is a bit tricky to observe the small-time behaviour numerically. The original "conductivity" in Claim~(\ref{claim:heatConduction}) is given by $\rho(z) = \sigma \left( 1 - \frac{k_0}{\sigma} z \right)^2 \Phi^{-1}(z)$ and we have replaced it by $\sigma (z/C_N)^{1/N} $, ignoring the boundary condition at $z=1$. To do this, the drive $w$ should be small (such that the conductivity is dominated by $\sigma (z/C_N)^{1/N}$, and the duration $T$ is naturally small), $\sigma$ should be small (such that the conductivity is naturally small, and a particle starting at small $z_0$ has small chance of "diffusing" to $z=1$), and $\frac{k_0}{\sigma}$ should be small (such that the term $\left( 1 - \frac{k_0}{\sigma} z \right)^2$ is unimportant). 

\subsubsection{Numerical Examples}

We shall study two systems,

\begin{enumerate}
\item System S1:$\Phi(\dot{h}) = \dot{h}^2$, $\Phi^{-1}(y) = y^{1/2}$, $k_0=1.0$, $\sigma=1.0$, under driving $w=0.01$. (Thus, $N=2$)
\item System S2:$\Phi(\dot{h}) = \dot{h}^4$, $\Phi^{-1}(y) = y^{1/4}$, $k_0=1.0$, $\sigma=1.0$, under driving $w=0.01$. (Thus, $N=4$)
\end{enumerate}

\begin{figure}
\centering
\subfigure[C.D.F.]{\includegraphics[scale=0.3]{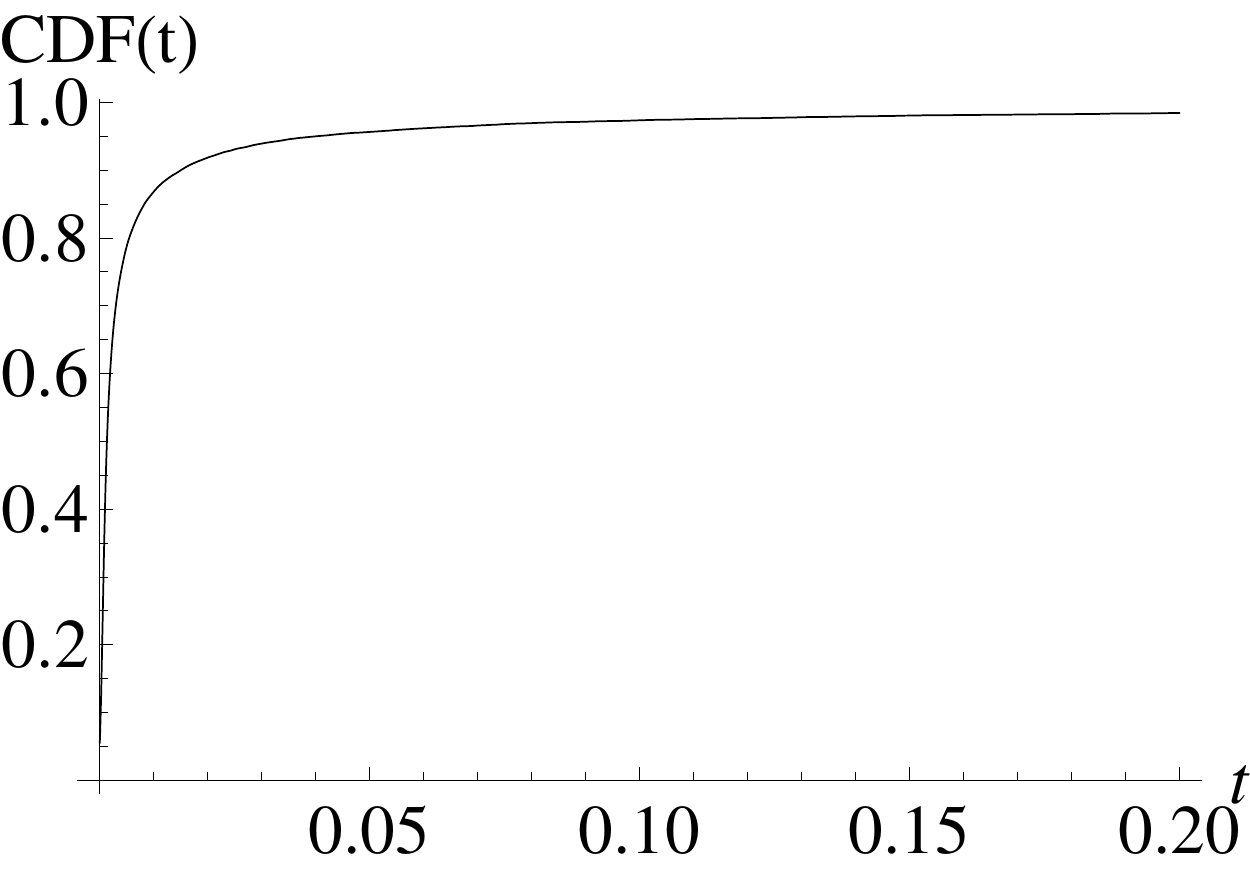}\label{fig:plotCDForder2}}\qquad
\subfigure[P.D.F.]{\includegraphics[scale=0.3]{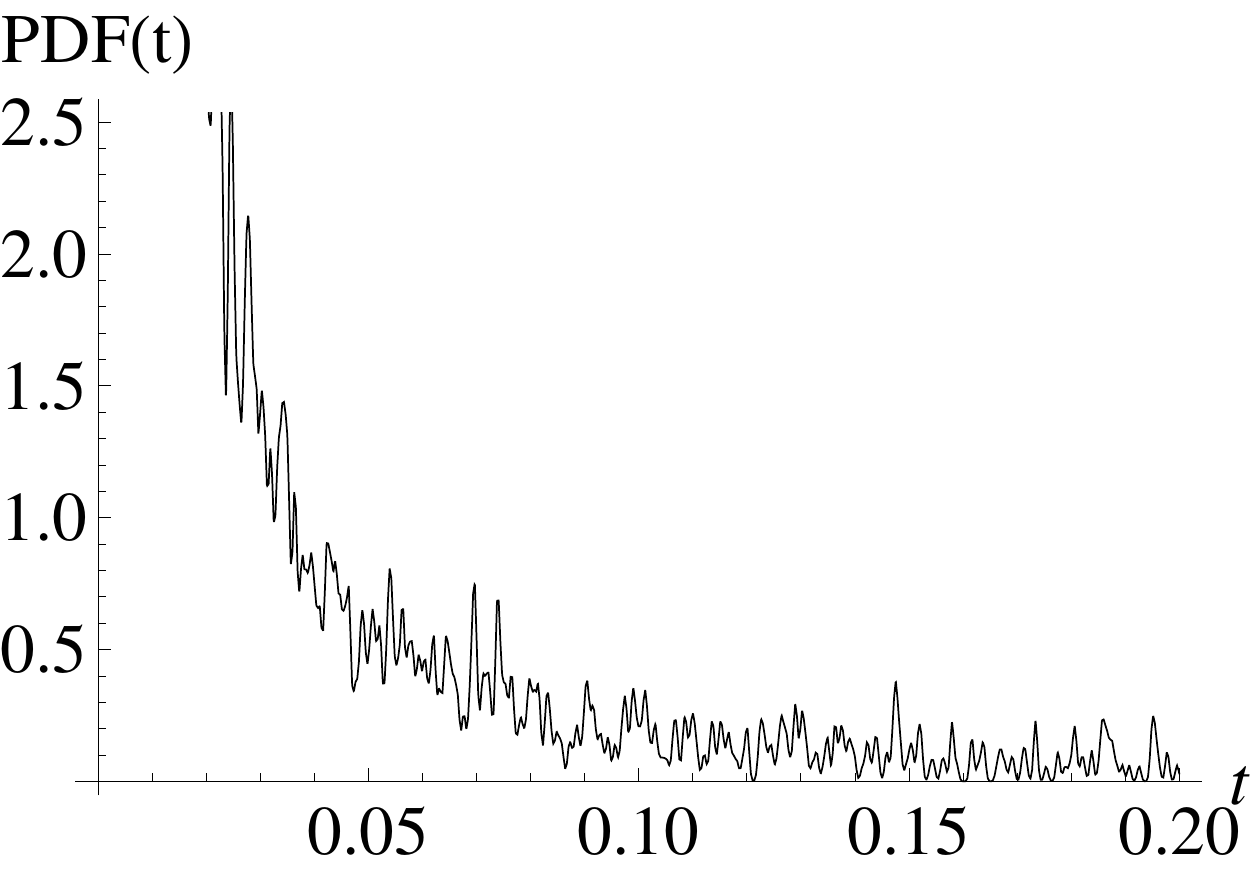}\label{fig:plotPDForder2}}\\
\subfigure[Log-log Plot of P.D.F.]{\includegraphics[scale=0.6]{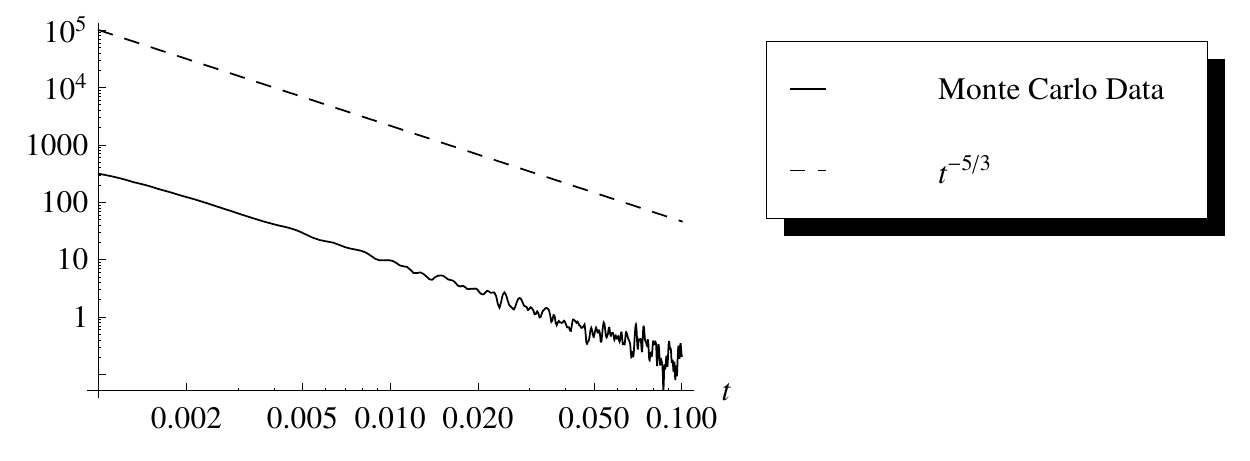}\label{fig:plotOrder2}}%
\caption{Fig.~\ref{fig:plotCDForder2} shows the cumulative distribution function of duration $T$ for System S1. 20000 paths are sampled in the Monte-Carlo study. The figure shows that the support of the probability density function is roughly within the interval $[0,0.1]$. Fig.~\ref{fig:plotPDForder2} shows the corresponding probability distribution function for System S1. Fig.~\ref{fig:plotOrder2} gives the log-log plot of the P.D.F. of duration $T$ for System S1. Shown also is the function $t^{-1 - N/(2 N+1)}$, where $N=2$, for comparison}
\label{fig:system1}
\end{figure}

\begin{figure}
\centering
\subfigure[C.D.F.]{\includegraphics[scale=0.3]{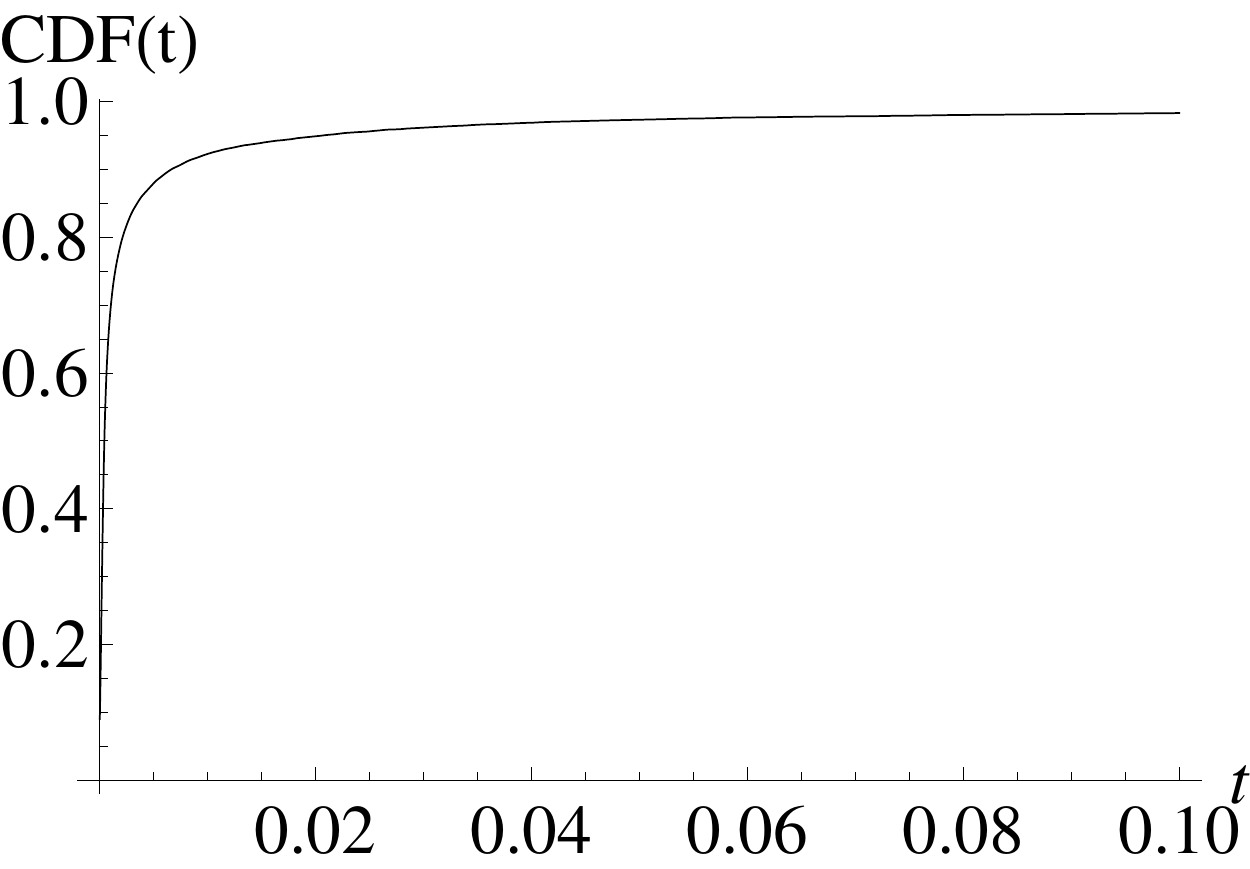}\label{fig:plotCDForder4}}\qquad
\subfigure[P.D.F.]{\includegraphics[scale=0.3]{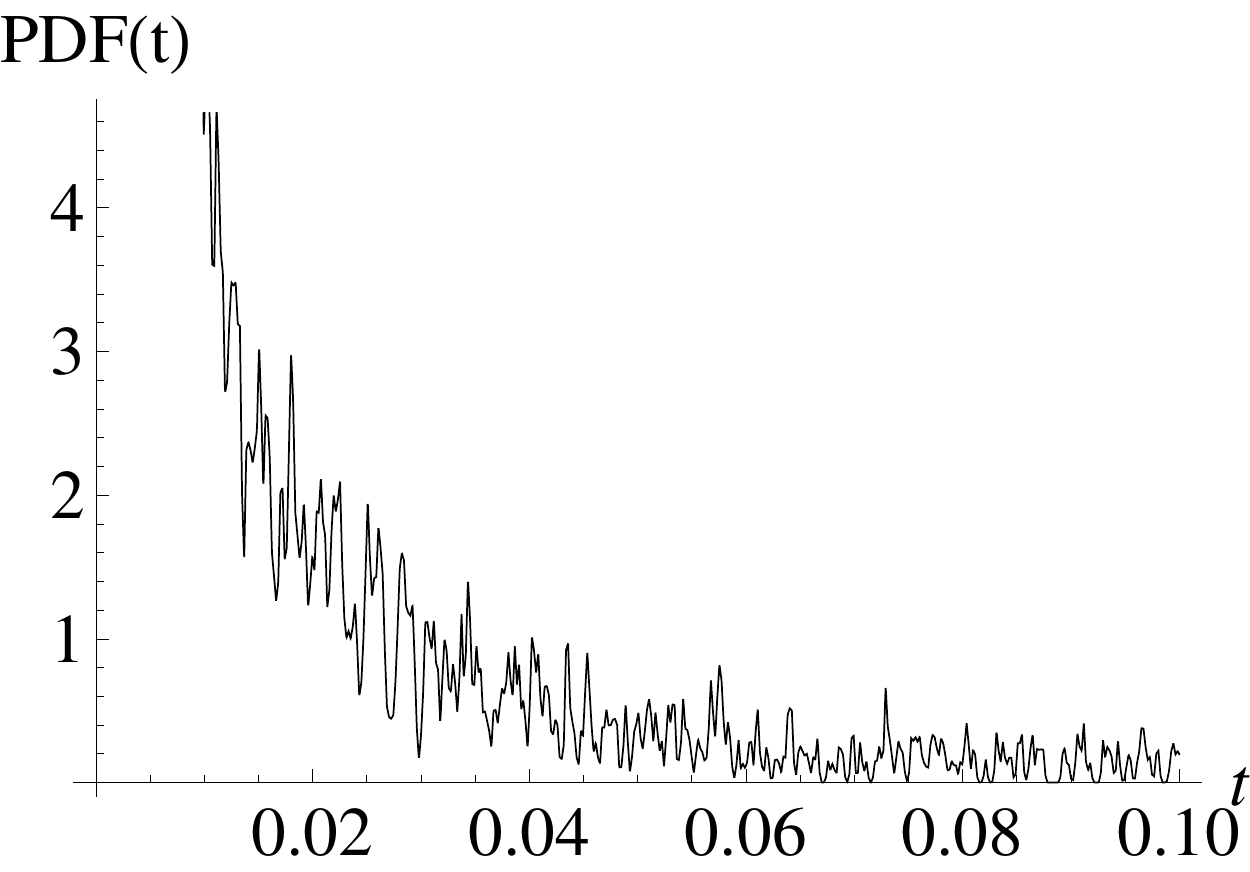}\label{fig:plotPDForder4}}\\
\subfigure[Log-log Plot of P.D.F.]{\includegraphics[scale=0.6]{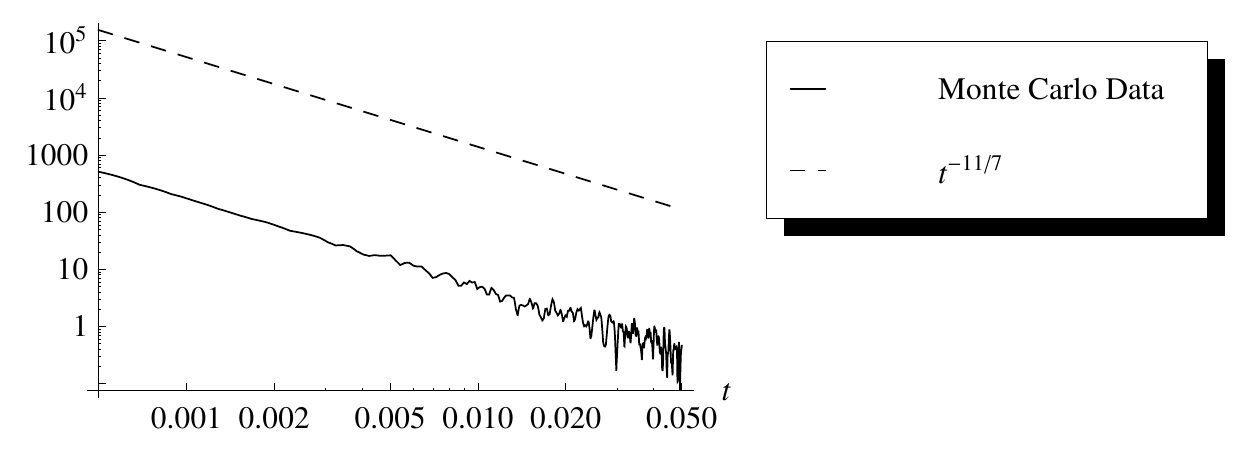}\label{fig:plotOrder4}}%
\caption{Fig.~\ref{fig:plotCDForder4} shows the cumulative distribution function of duration $T$ for System S2. 20000 paths are sampled in the Monte-Carlo study. The figure shows that the support of the probability density function is roughly within the interval $[0,0.1]$. Fig.~\ref{fig:plotPDForder4} shows the corresponding probability distribution function for System S2. Fig.~\ref{fig:plotOrder4} gives the log-log plot of the P.D.F. of duration $T$ for System S2. Shown also is the function $t^{-1 - N/(2 N+1)}$, where $N=4$, for comparison}
\label{fig:system1}
\end{figure}

Figures~(\ref{fig:plotCDForder2}) and ~(\ref{fig:plotPDForder2}) show the cumulative distribution function (C.D.F.) and probability density function (P.D.F.) of durations from a Monte-Carlo simulation for System S1. Figure~(\ref{fig:plotOrder2}) shows the log-log plot of P.D.F. of durations. It is shown that the data are fitted quite well with the function $t^{-1 - N/(2 N+1)}$, where $N=2$. Figures~(\ref{fig:plotCDForder4}), ~(\ref{fig:plotPDForder4}) and ~(\ref{fig:plotOrder4}) show similiar information for System S2. 







\section{Conclusion}
\label{sec:conclusion}

In this paper, we noted some general features common to systems described by the ABBM model with a velocity-dependent dissipative coefficient $\eta$. Some questions arise. First and foremost, is there any physical system which can be described, at least phenomenologically, by our model with a non-trivial $\eta$? For example, is there any physical system exhibiting the behaviour $\Phi(\dot{h}) \sim \dot{h}^2$? We mention that for large Reynolds number the drag is proportional to the squared velocity rather the velocity (See the chapter "The drag crisis" in \cite{landau2013fluid}), which may or may not be relevant to the current problem. Moreover, can equations similiar to ours be derived based on microscopic considerations? Although a quick answer may be negative due to the high non-linearity introduced by a velocity-dependent $\eta$, one cannot rule out the possibility that it arises as an "effective" (though not necessarily "exact") description of the picture. We leave these questions to future studies. 

\nocite{*}

\bibliography{apssamp}

\end{document}